\documentclass[aps,amsmath,amssymb,twocolumn]{revtex4}
\UseRawInputEncoding
\pdfoutput=1
\usepackage[english]{babel}
\usepackage{graphicx}% Include figure files
\usepackage{dcolumn}% Align table columns on decimal point
\usepackage{bm}% bold math
\usepackage{dsfont}
\usepackage{enumerate}
\usepackage{color}
\usepackage{mathrsfs}
\usepackage{braket}
\usepackage{amssymb,amsmath,amsthm}
 
 \newtheorem{theorem}{Theorem}
 
\begin{document}
\title{The applications of EPR steering in quantum teleportation for two- or three-qubit system}
\author{Yi Fan}
\author{Liang Qiu}\email{lqiu@cumt.edu.cn}
\author{Chuanlei Jia}\email{jiachl@cumt.edu.cn}
\affiliation{School of Materials Science and Physics, China University of Mining and Technology, Xuzhou 221116, People's Republic of China}
\author{Yiyue Gu}\email{solargu@163.com}
\affiliation{Xuzhou No.1 People's Hospital, Xuzhou 221116, People's Republic of China}
%%%%%%%%%%%%%%%%

\begin{abstract}
EPR steering is an important quantum resource in quantum information and computation. In this paper, its applications in quantum teleportation are investigated. First of all, the upper bound of the average teleportation fidelity based on the EPR steering is derived. When the receiver can only perform the identity or the Pauli rotation operations, the X-type states which violate the three-setting linear steering inequality could be used for teleportation. In the end, the steering observables and the average teleportation fidelities of the two-qubit reduced states for three-qubit pure states maintain the same ordering. The complementary relations between the steerable observables and the average teleportation fidelities for three-qubit pure states are also established.
\end{abstract}

\maketitle

\section{Introduction}\label{sec1}
Quantum steering, given by Schr\"{o}dinger in 1935 \cite{Schrodinger1935,Schrodinger1936} within the context of the Einstein-Podolsky-Rosen (EPR) paradox, captures the ability that the local measurements of Alice can remotely steer the state of Bob if an entangled state is shared between them. However, there was no significant progress in the research field of EPR steering until 2007, when an operational interpretation of quantum steering was given in the form of a task \cite{Wiseman,Jones}. In the task, Alice wanted to persuade Bob that they shared an entangled state even if the latter didn't trust her. Based on the task, the hierarchy among three quantum correlations, i.e., Bell nonlocality, EPR steering and entanglement, is established. The set of EPR steerable states is a subset of entangled states and a superset of Bell nonlocal states. Apart from having foundational significance, EPR steering has wide applications ranging from one-sided device-independent quantum key distribution \cite{Branciard2012}, secure quantum communication \cite{Reid2013}, secure quantum teleportation \cite{Reid2013,He}, detecting bound entangled states \cite{Moroder}, subchannel discrimination \cite{Piani}, randomness generation \cite{Law,Passaro,Skrzypczyk,Coyle} and self-testing of pure entangled states \cite{Supic,Gheorghiu,Goswami}. Using correlations, state assemblages and full information, a range of steering criteria are developed to quantify the amount of EPR steering \cite{Jevtic2014,Kogias}, such as the steerable weight \cite{Pusey,Skrzypczyk2014}, the steering robustness \cite{Piani}. Beyond these steering criteria, several inequalities in terms of the sufficient conditions of steerability are used to detect steerability \cite{Reid,Cavalcanti,Walborn,Schneeloch,Chen,Kogias,Cavalcanti2015,Zukowski,Jevtic,Costa,Uola,Dai}.

Quantum teleportation \cite{Bennett} is a fundamental and important protocol in quantum information science, by which a sender transfers an unknown quantum state to a receiver using shared entanglement. The sender and the receiver, separated spatially, are only allowed to perform local operations and communicate classically \cite{Bennett1996}. A class of maximally entangled states is chosen in the original teleportation protocol. However, the available states are typically mixed entangled in practice, and they could be used for imperfect teleportation \cite{Popescu}. The resemblance of two quantum states and the properties of quantum teleportation could be quantified by the fidelity, which measures the overlap of the state to be teleported and the output state. Because the state to be teleported is generally unknown, the average fidelity over all input states is used as the standard figure of merit for quantum teleportation \cite{Horodecki,Gisin,Horodecki1996,Horodecki1999,Verstraete,Ghosal}. The average fidelity can indicate whether the entangled mixed states can be used for teleportation.

As quantum correlations are important quantum resources in quantum information and computation, their applications have attracted great interest. Particularly, the relationship between quantum correlations and teleportation has been deeply investigated. Horodecki et al. demonstrated that two-qubit states violating the Bell-CHSH inequality were useful for teleportation if the sender used only the Bell basis in her measurement and the receiver could apply any unitary transformation \cite{Horodecki}. When Bob was restricted to perform the identity or the Pauli rotation operations, Hu et al. found that the X-type states which violated the Bell-CHSH inequality could be used for teleportation \cite{Hu}. The tight lower and upper bounds on the maximal singlet fraction for a given amount of entanglement which was measured by concurrence or negativity were derived \cite{Verstraete2002}. For three-qubit pure states, the partial tangle that represented the residual two-qubit entanglement was closely related to the fidelity of teleportation \cite{Lee}. By connecting geometrical discord to the teleportation fidelity, an operational meaning of it was provided \cite{Adhikari}. Very recently, it was found that two-qubit states violating the three-setting linear steering inequality were useful for teleportation, and the optimal two-qubit states for teleportation for a given value of EPR steering were presented \cite{Fan}. However, it should be noted that the upper bound on the fidelity of quantum teleportation based on EPR steering has not been given yet. It is also interesting to consider whether the X-type states violating the three-setting linear steering inequality could be used for teleportation when the receiver is only allowed to perform the identity or the Pauli rotation operations. Furthermore, the relationship between EPR steering and teleportation fidelity for three-qubit states deserves investigation.

In this paper, we investigate the applications of EPR steering in quantum teleportation. The rest of the paper is arranged as follows. In Sec. II, the definitions of the average teleportation fidelity and EPR steering are reviewed. In Sec. III, the upper bound of the average teleportation fidelity based on EPR steering is derived. In Sec. IV, we show that the X-type states violating the three-setting linear steering inequality are useful for teleportation when the receiver can only apply identity or Pauli rotation operations. Subsequently, we relate EPR steering with the average teleportation fidelity for three-qubit pure states in Sec. V. A conclusion and discussion part is given in Sec. VI.

\section{Preliminaries}\label{sec2}
In this section, the definitions of the average teleportation fidelity and EPR steering are reviewed.

The Hilbert-Schmidt decomposition of a two-qubit density matrix $\rho$ can be represented as
\begin{equation}
\rho=\frac{1}{4}\left(I\otimes I+\vec{a}\cdot\vec{\sigma}\otimes I+I\otimes\vec{b}\cdot\vec{\sigma}+\sum\limits_{i,j=1}^{3}r_{i,j}\sigma_{i}\otimes\sigma_{j}\right),\label{eq1}
\end{equation}
where $\vec{a}$ and $\vec{b}$ are vectors in $\mathbb{R}^3$, $\vec{a}(\vec{b})\cdot\vec{\sigma}=\sum_{i=1}^3a_{i}(b_{i})\sigma_{i}$ with $\sigma_{i}$ being the Pauli matrix. $r_{i,j}={\rm Tr}(\rho \sigma_{i}\otimes\sigma_{j})$ are elements of the correlation matrix $R$.

The average teleportation fidelity measures the efficiency of a two-qubit state $\rho$ used for teleportation \cite{Horodecki1996}
\begin{equation}
f(\rho)=\int_{\mathcal{S}}\sum\limits_{k}p_{k}{\rm Tr}(\rho_{k}\rho_{\phi}) dM(\phi),
\end{equation}
where $\rho_{k}$ is the output state when Alice's measurement outcome is $k$. $\rho_{\phi}$ is the input pure state. The quantity ${\rm Tr}(\rho_{k}\rho_{\phi})$ measures how the resulting state is similar to the input one, and it is averaged over the probabilities of outcomes $p_{k}$. The integral is over a uniform distribution on the Bloch sphere $\mathcal{S}$ of all input states $\rho_{\phi}$. In Ref. \cite{Horodecki}, it was shown that the average teleportation fidelity could be given as follows if a two-qubit state could be used for standard teleportation
\begin{equation}
f(\rho)=\frac{1}{2}\left[1+\frac{N(\rho)}{3}\right],\label{eq3}
\end{equation}
where $N(\rho)=\sum_{i=1}^3\sqrt{u_{i}}$ with $u_{i}$'s being the eigenvalues of the matrix $R^{T}R$. Here, the superscript $T$ denotes the transpose of the matrix.

Three-setting linear steering inequality, which is used to detect the steerability of a state \cite{Cavalcanti} and is also considered as a very useful criterion for the quantification of EPR steering \cite{Costa}. For a two-qubit system, it can be expressed as \cite{Cavalcanti}
\begin{equation}
\frac{1}{\sqrt{3}}\left|\sum\limits_{i=1}^3{\rm Tr}(A_{i}\otimes B_{i}\rho)\right|\leq1,\label{eq4}
\end{equation}
where the Hermitian operators are expressed as $A_{i}=\vec{a}_{i}\cdot\vec{\sigma}$ and $B_{i}=\vec{b}_{i}\otimes\vec{\sigma}$ ($i=1,2,3$), respectively. They act on the corresponding subsystems $A$ and $B$. $\vec{a}_{i}, \vec{b}_{i}\in\mathbb{R}^3$ are two unit vectors, and $\vec{b}_{i}$ are orthogonal vectors. If the inequality given in Eq. (\ref{eq4}) is violated, the two-qubit state $\rho$ is steerable. Generally, the steering observable is given by the maximum violation
\begin{equation}
S(\rho)=\max\limits_{\{A_{i},B_{i}\}}\frac{1}{\sqrt{3}}\left|\sum\limits_{i=1}^3{\rm Tr}(A_{i}\otimes B_{i}\rho)\right|=\sqrt{{\rm Tr}(R^{T}R)}.\label{eq5}
\end{equation}

Based on the Hilbert-Schmidt decomposition of a two-qubit state given in Eq. (\ref{eq1}), the explicit expression of the steering observable is
\begin{equation}
S(\rho)=\sqrt{\sum_{i=1}^{3}u_{i}}.
\end{equation}

Obviously, the states violate the inequality and are steerable if $S(\rho)>1$. However, it should be noted that it is just a sufficient criterion to check steerability because there exist steerable states satisfying $S(\rho)\leq1$.

\section{Upper bound of the average teleportation fidelity}
In Ref. \cite{Verstraete2002}, the authors derived the tight upper and lower bounds of the singlet fraction for given values of concurrence or negativity. The singlet fraction defined for a two-qubit state $\rho$ was $F(\rho)={\rm \max}\langle\Psi|\rho|\Psi\rangle$ with the maximum being taken over all two-qubit maximally entangled states $|\Psi\rangle$. As known to all, the relation between the singlet fraction and the average teleportation fidelity was $f(\rho)=[2F(\rho)+1]/3$. Therefore, it also implied that the upper and lower bounds of the teleportation fidelity for given values of concurrence or negativity were derived \cite{Verstraete2002}. Motivated by these results, the upper bound of the average teleportation fidelity for given values of EPR steering is investigated in this section. Here, the EPR steering is measured by the steering observable and we just consider the case that the two-qubit state is useful for standard teleportation protocol, which implies that the average teleportation fidelity is always larger than $2/3$.

\begin{theorem}
Given a two-qubit state $\rho$ with the steering observable $S(\rho)$, the upper bound of the average teleportation fidelity is
\begin{equation}
f(\rho)\leq\frac{1}{2}\left[1+\frac{S(\rho)}{\sqrt{3}}\right].
\end{equation}
\end{theorem}

\begin{proof}
Because $(\sum_{i}\sqrt{u_{i}})^2\leq3(\sum_{i}u_{i})$, $N(\rho)=\sum_{i}\sqrt{u_{i}}$ is always smaller than $\sqrt{3(\sum_{i}u_{i})}=\sqrt{3}S(\rho)$. Therefore, the upper bound of the average teleportation fidelity is obtained.
\end{proof}

%From the lower bound, it is found that the average teleportation fidelity is always larger than $2/3$ and two-qubits states are useful for teleportation if two-qubit states violate the three-setting linear steering inequality.

The upper bound of the average teleportation fidelity is saturated iff $u_{1}=u_{2}=u_{3}=S(\rho)/\sqrt{3}$. This result is consistent with the conclusion given in Ref. \cite{Fan}.

In Ref. \cite{Verstraete2002}, the tight upper bound on average teleportation fidelity for given value of concurrence is $[2{\rm \max}\{C,(1+C)/4\}+1]/3\leq f(\rho)\leq(2+C)/3$, in which $C$ is the value of concurrence of the state $\rho$. Now, we can give some examples to indicate that the bound based on the steering observable could be tighter than that based on concurrence.

Ishizaka et al. \cite{Ishizaka,Verstraete2001} introduced a class of two-qubit maximally entangled states (MEMS) \cite{Paulson}
\begin{equation}
\rho_{\rm M}=\lambda_{1}|\psi^{-}\rangle\langle\psi^{-}|+\lambda_{2}|00\rangle\langle00|
+\lambda_{3}|\psi^{+}\rangle\langle\psi^{+}|+\lambda_{4}|11\rangle\langle11|,
\end{equation}
where $|\psi^{\pm}\rangle=(|01\rangle\pm|10\rangle)/\sqrt{2}$ were the maximally entangled states. $\lambda_{i}$ ($i=1,\cdots,4$) were the eigenvalues of the state $\rho_{\rm M}$.

The MEMS state reduces to the state with rank $3$ when one chooses $\lambda_{1}=(1+2p)/3,\ \lambda_{2}=\lambda_{3}=(1-p)/3,$ and $\lambda_{4}=0$, where $p\in[0,1]$. Through straightforward calculation, the steering observable and concurrence of the state are $\sqrt{(1+4p+22p^2)/9}$ and $p$, respectively. The condition $p>(3\sqrt{5}-1)/11$ will ensure the steering observable $S(\rho_{\rm M})$ being larger than $1$, i.e., the state is steerable. The condition can also ensure the state being useful for teleportation because the average teleportation fidelity is larger than $2/3$ \cite{Fan}. When it comes to the upper bound of the teleportation fidelity, we compare the bound based on the steering observable with that based on concurrence in Fig. 1(a). From the figure, it is found that the upper bound of the average teleportation fidelity based on the steering observable is always smaller than that based on concurrence. In other words, a tighter bound is obtained for this state.

The MEMS state reduces to the state with rank $2$ when taking $\lambda_{1}=(1+p)/2,\ \lambda_{2}=(1-p)/2,\ \lambda_{3}=\lambda_{4}=0$ into consideration. For this state, the steering observable and concurrence are $\sqrt{(1+2p+3p^2)/2}$ and  $(1+p)/2$, respectively. Through simple calculation, it is easy to find that the state is steerable with $p\in(1/3,1]$ and, in this region, the state is also useful for teleportation \cite{Fan}. For this state, the upper bound of the average teleportation fidelity as a function of the state parameter $p$ is plotted in Fig. 1(b), and one can note that the upper bound based on the steering observable is always smaller than that based on concurrence. Thus, the upper bound based on the steering observable is much tighter for this state.
\begin{figure}[h]%[tbp]
\begin{center}
\includegraphics[scale=0.5]{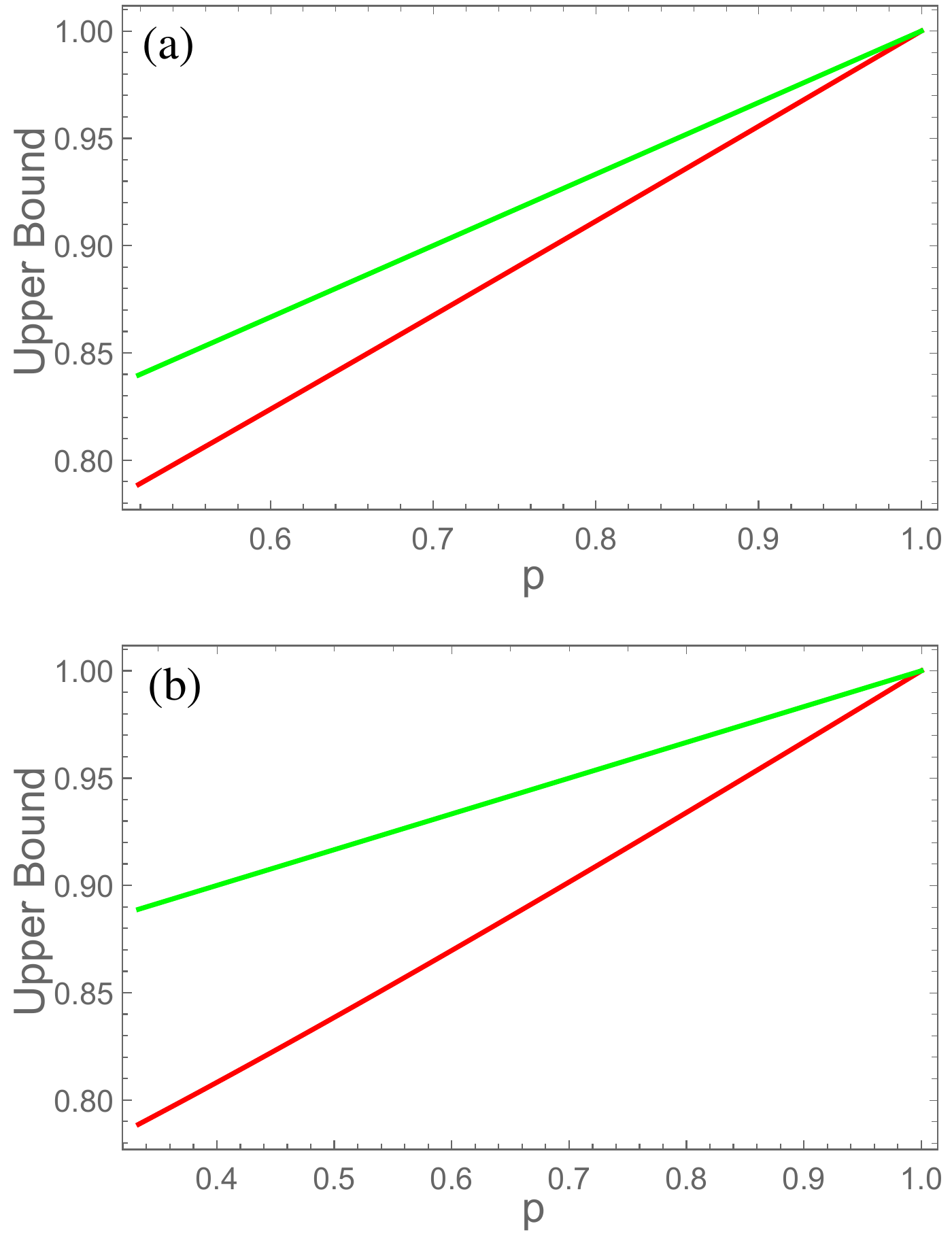}
\end{center}
\caption{The upper bound of the average teleportation fidelity as a function of the state parameter $p$ for the MEMS state with (a) $\lambda_{1}=(1+2p)/3,\ \lambda_{2}=\lambda_{3}=(1-p)/3,$ and $\lambda_{4}=0$, (b) $\lambda_{1}=(1+p)/2,\ \lambda_{2}=(1-p)/2$, and $\lambda_{3}=\lambda_{4}=0$. The red line labels the upper bound of the average teleportation fidelity based on the steering observable, while the green line labels that based on concurrence.}
\end{figure}

Note that the bound of the average teleportation fidelity based on the steering observable is not always tighter than that based on concurrence. For example, the upper bound based on the steering observable equals to that based on concurrence for the optimal state with a definite value of concurrence \cite{Ghosal} or the optimal state with a definite value of the steering observable \cite{Fan}. What's worse, the upper bound obtained in this paper could be larger than that based on concurrence for two-qubit pure state $\alpha|00\rangle\pm\beta|11\rangle$ with the complex parameters $\alpha$, $\beta$ satisfying $|\alpha|^2+|\beta|^2=1$. The physical reasons behind these phenomena may be that the steering observable $S(\rho)$ cannot detect all the steerable states.

\section{X-type states violating the three-setting linear steering inequality being useful for restricted teleportation}
In Ref. \cite{Horodecki}, Horodecki et al. proved that two-qubit states violating the Bell-CHSH inequality were useful for teleportation when the receiver Bob could apply any unitary transformation. However, some unitary transformations may be difficult to perform in real circumstances. In Ref. \cite{Hu}, the authors found that the X-type states violating the Bell-CHSH inequality could be used for teleportation even if Bob was restricted to apply the identity or the Pauli rotation operations. Recently, it was shown that the states violating the three-setting linear steering inequality was useful for teleportation no matter what unitary transformations Bob performed \cite{Fan}. Therefore, a question whether the X-type states violating the three-setting linear steering inequality could be used for teleportation when Bob can only apply the identity or the Pauli rotation operations occurs. An affirmative answer to this question is given in this section.

The X-type state shared between the sender Alice and the receiver Bob could be expressed as follows
\begin{equation}
\rho_{X}=\left(
\begin{array}{llll}
a&0&0&w\\
0&b&z&0\\
a&z^{*}&c&0\\
w^{*}&0&0&d
\end{array}
\right).
\end{equation}
The diagonal elements $a,b,c,d$ are non-negative and satisfy $a+b+c+d=1$ due to the normalization of the density matrix. The anti-diagonal elements $w,z$ are complex and satisfy $|w|^2\leq ad, |z|^2\leq bc$ due to the positive semi-definiteness of the density matrix.

When the receiver Bob can only apply the identity or the Pauli rotation operations, the average teleportation fidelity for the X-type states can be given as \cite{Horodecki1999,Yeo}
\begin{equation}
f(\rho_{X})=\frac{2\mathcal{F}(\rho_{X})+1}{3},
\end{equation}
where the fully entangled fraction \cite{Bowen,Albeveio} is expressed as $\mathcal{F}(\rho_{X})=\max\{\chi_{0},\chi_{1},\chi_2,\chi_3\}$ with $\chi_{0,3}=(a+d\pm2|w|)/2$ and $\chi_{1,2}=(b+c\pm2|z|)/2$. Obviously, $f(\rho_{X})$ is actually determined by the quantity $\chi_{0}$ and $\chi_{1}$ due to the fact $\chi_0\geq\chi_3, \chi_{1}\geq\chi_{2}$.

Now, we discuss the relation between EPR steering and the average teleportation fidelity when Bob can only apply the identity or the Pauli rotation operations. Because $u_{1,2}(\rho_{X})=4(|w|\pm|z|)^2$ and $u_{3}(\rho_{X})=(a+d-b-c)^2$, the steering observable of the X-type state is $S(\rho_{X})=\sqrt{\sum_{i}u_{i}(\rho_{X})}=\sqrt{4(|w|^2+|z|^2)+(a+d-b-c)^2}$.

First, we will show that $\chi_{0}<1/2$ and $\chi_{1}<1/2$ will not be satisfied simultaneously if the X-type states violate the three-setting linear steering inequality, i.e., the steering observable $S(\rho_{X})>1$. For $S(\rho_{X})>1$, one has $4(|w|^2+|z|^2)+(a+d-b-c)^2>1$, which will give
\begin{equation}
\begin{array}{ll}
4(|w|^2+|z|^2)&>1-(a+d-b-c)^2\\
&=(a+b+c+d)^2-(a+d-b-c)^2\\
&=4(a+d)(b+c).
\end{array}
\nonumber
\end{equation}
Thus, the inequality $|w|^2+|z|^2>(a+d)(b+c)$ is obtained when the X-type states violate the three-setting linear steering inequality.

On the other hand, if $\chi_{0}<1/2$ and $\chi_{1}<1/2$ are satisfied simultaneously, one will get $2|w|<b+c$ and $2|z|<a+d$ from the expressions of $\chi_{0}$ and $\chi_{1}$. From the positive semi-definiteness, one has $2|w|\leq2\sqrt{ad}\leq a+d$ and $2|z|\leq2\sqrt{bc}\leq b+c$. Combining them together, one will find $|w|+|z|<b+c$ and $|w|+|z|<a+d$. Therefore, the result $|w|^2+|z|^2<(|w|+|z|)^2<(a+d)(b+c)$ is obtained, which is in contradiction with the result getting from $S(\rho_{X})>1$. Therefore, one can see that the X-type states will not violate the three-setting linear steering inequality if $\chi_{0}<1/2$ and $\chi_{1}<1/2$.

Second, $\chi_{0}>1/2$ and $\chi_{1}>1/2$ will not also be satisfied simultaneously \cite{Hu}. If these two inequalities are satisfied simultaneously, one will get $\chi_{0}+\chi_{1}>1$, from which $|w|+|z|>1/2$ is obtained.  This result is in contradiction with $|w|+|z|<(a+b+c+d)/2=1/2$, which can be obtained from $|w|<(a+d)/2$ and $|z|<(b+c)/2$ due to the positive semi-definiteness of the X-type states.

Based on the above discussion, one can get the following result.

\begin{theorem}The X-type states violating the three-setting linear steering inequality are useful for teleportation even if the receiver can only perform the identity or the Pauli rotation operations.
\end{theorem}

\section{Relating EPR steering with teleportation fidelity for three-qubit pure states}
In the above two sections, we consider the relationship between EPR steering and average teleportation fidelity in two-qubit system, and now we extend the investigation of the relationship to multipartite system, particularly, three-qubit system. The three-qubit pure state can be expressed as
\begin{equation}
|\psi\rangle_{123}=\alpha_0|000\rangle+\alpha_{1}e^{i\theta}|100\rangle+\alpha_{2}|101\rangle+\alpha_{3}|110\rangle+\alpha_{4}|111\rangle,
\end{equation}
where $\alpha_{i}\geq0$, $\sum_{i}\alpha_{i}^2=1$, and $\theta\in[0,\pi]$ is a phase.

The teleportation scheme over the three-qubit state is described as follows \cite{Lee}. First, Alice makes one-qubit orthogonal measurement on the subsystem $i$. Second, Bob makes two-qubit orthogonal measurements on the qubit whose state is to be teleported and the subsystem $j$. In the end, Charlie applies a corresponding unitary operation on the subsystem $k$ after receiving three bits classical information on the results of the two measurements. Here, $i,j,k=1,2,3$ and $i\neq j\neq k$. The above scheme is actually a standard teleportation scheme based on the reduced state $\rho_{jk}$ after the measurement on the subsystem $i$.

In Ref. \cite{Lee}, the authors introduced the definition of the partial tangle to relate it with the average teleportation fidelity, and the partial tangle was defined as
\begin{equation}
\tau_{ij}=\sqrt{C_{i(jk)}^2-C_{ik}^2},
\end{equation}
where $C_{i(jk)}=2\sqrt{\det{[{\rm Tr_{jk}}(|\psi\rangle_{ijk}\langle\psi|)]}}$. $C_{ik}=C[{\rm Tr_j}(|\psi\rangle_{ijk}\langle\psi|)]$ was the concurrence of the reduced two-qubit state $\rho_{ik}$.
Through straightforward calculation, the partial tangle for three-qubit pure states can be explicitly given as \cite{Lee}
\begin{equation}
\begin{array}{c}
\tau_{12}=2\alpha_0\sqrt{\alpha_{3}^2+\alpha_{4}^2},\\
\tau_{23}=2\sqrt{\alpha_0^2\alpha_4^2+\alpha_1^2\alpha_4^2+\alpha_{2}^2\alpha_{3}^2-2\alpha_1\alpha_2\alpha_3\alpha_4\cos{\theta}},\\
\tau_{13}=2\alpha_0\sqrt{\alpha_{2}^2+\alpha_{4}^2}.\label{eq13}
\end{array}
\end{equation}
The relationship between the partial tangle and the average teleportation fidelity for the corresponding subsystem had been explicitly expressed as \cite{Lee}
\begin{equation}
\tau_{ij}=3f(\rho_{ij})-2.\label{eq14}
\end{equation}

The steering observables for the corresponding subsystems of three-qubit pure states can be calculated as \cite{Paul}
\begin{equation}
\begin{array}{ll}
S(\rho_{12})=&\left\{1+8\alpha_{0}^2\alpha_{3}^2-4\alpha_{0}^2\alpha_{2}^2-4\alpha_{1}^2\alpha_{4}^2-4\alpha_{2}^2\alpha_{3}^2\right.\\
&\left.+8\alpha_{1}\alpha_{2}\alpha_{3}\alpha_{4}\cos{\theta}\right\}^{\frac{1}{2}},\\
S(\rho_{13})=&\left\{1+8\alpha_{0}^2\alpha_{2}^2-4\alpha_{0}^2\alpha_{3}^2-4\alpha_{1}^2\alpha_{4}^2-4\alpha_{2}^2\alpha_{3}^2\right.\\
&\left.+8\alpha_{1}\alpha_{2}\alpha_{3}\alpha_{4}\cos{\theta}\right\}^{\frac{1}{2}},\\
S(\rho_{23})=&\left\{1-4\alpha_{0}^2\alpha_{2}^2-4\alpha_{0}^2\alpha_{3}^2+8\alpha_{1}^2\alpha_{4}^2+8\alpha_{2}^2\alpha_{3}^2\right.\\
&\left.-16\alpha_{1}\alpha_{2}\alpha_{3}\alpha_{4}\cos{\theta}\right\}^{\frac{1}{2}}.\label{eq15}
\end{array}
\end{equation}

The first result between the steering observable and the average teleportation fidelity for three-qubit pure states is given as follows.
\begin{theorem}
The triples $(f(\rho_{12}),f(\rho_{13}),f(\rho_{23}))$ of three reduced states obtained from a pure three-qubit state and $(S(\rho_{12}),S(\rho_{13}),S(\rho_{23}))$ maintain the same ordering, i.e., $f(\rho_{12})>f(\rho_{13})>f(\rho_{23})$ iff $S(\rho_{12})>S(\rho_{13})>S(\rho_{23})$.
\end{theorem}

\begin{proof}
From Eqs. (\ref{eq13}) and (\ref{eq15}), one will get $S^2(\rho_{12})-S^2(\rho_{13})=12\alpha_{0}^2(\alpha_{3}^2-\alpha_{2}^2)$ and $\tau_{12}^2-\tau_{13}^2=4\alpha_{0}^2(\alpha_{3}^2-\alpha_{2}^2)$. Thus, $\tau_{12}>\tau_{13}$ if $S(\rho_{12})>S(\rho_{13})$. Due to the fact $f(\rho_{ij})\geq2/3$ for three-qubit pure states and the relation given in Eq. (\ref{eq14}), one will get the result that $f(\rho_{12})>f(\rho_{13})$ iff $S(\rho_{12})>S(\rho_{13})$. Similarly, one can prove that $f(\rho_{12})>f(\rho_{23})$ iff $S(\rho_{12})>S(\rho_{23})$, and $f(\rho_{13})>f(\rho_{23})$ iff $S(\rho_{13})>S(\rho_{23})$. Therefore, the proof is completed.
\end{proof}

For three-qubit pure states, the result concludes that the larger the EPR steering of the reduced state is, the greater the average teleportation fidelity of the reduced state is.

The second result between the steering observable and the average teleportation fidelity for three-qubit pure states is that there is a complementary relation.
\begin{theorem}
For a three-qubit pure state, the steering observables of two reduced two-qubit states and the average teleportation fidelity of the third reduced two-qubit state obey the following complementary relation
\begin{equation}
S^2(\rho_{ij})+S^2(\rho_{ik})+[3f(\rho_{jk})-2]^2\leq3.
\end{equation}
\end{theorem}

\begin{proof}
The proof is given as follows. From Eqs. (\ref{eq13}) and (\ref{eq15}), one can obtain
\begin{equation}
\begin{array}{l}
S^2(\rho_{12})+S^2(\rho_{13})+\tau_{23}^2\\
=2+4\alpha_{0}^2(\alpha_{2}^2+\alpha_{3}^2+\alpha_{4}^2)-4\alpha_{1}^2\alpha_{4}^2-4\alpha_{2}^2\alpha_{3}^2\\
+8\alpha_{1}\alpha_{2}\alpha_{3}\alpha_{4}\cos{\theta}\\
=2+4\alpha_{0}^2(1-\alpha_{0}^2-\alpha_{1}^2)-4(\alpha_{1}\alpha_{4}-\alpha_{2}\alpha_{4})^2\\
-8\alpha_{1}\alpha_{2}\alpha_{3}\alpha_{4}(1-\cos{\theta})\\
=3-(2\alpha_0^2-1)^2-4\alpha_0^2\alpha_1^2-4(\alpha_{1}\alpha_{4}-\alpha_{2}\alpha_{4})^2\\
-8\alpha_{1}\alpha_{2}\alpha_{3}\alpha_{4}(1-\cos{\theta})\\
\leq3,
\end{array}
\end{equation}
where the normalization of the three-qubit pure states, i.e., $\sum_{i}\alpha_{i}^2=1$, is used. Noting $\tau_{23}=3f(\rho_{23})-2$, one could obtain the complementary relationship $S^2(\rho_{12})+S^2(\rho_{13})+[3f(\rho_{23})-2]^2\leq3$. Similarly, one can also give $S^2(\rho_{12})+S^2(\rho_{23})+[3f(\rho_{13})-2]^2\leq3$ and $S^2(\rho_{13})+S^2(\rho_{23})+[3f(\rho_{12})-2]^2\leq3$. Thus, the proof of the second result is completed.
\end{proof}

From the second result, one can see that the steering observables of two reduced two-qubit states and the average teleportation fidelity of the third reduced two-qubit state are upper bounded. On the other hand, this is a trade-off relation among the steerable observables and the average teleportation fidelity for three-qubit pure states.

With the similar procedure, one can also obtain another complementary relation between the steering observable and the average teleportation fidelity.
\begin{theorem}
For a three-qubit pure state, the steering observable of a reduced two-qubit state and the average teleportation fidelities of the other two reduced two-qubit states obey the following complementary relation
\begin{equation}
S^2(\rho_{ij})+[3f(\rho_{ik})-2]^2+[3f(\rho_{jk})-2]^2\leq3.
\end{equation}
\end{theorem}

\begin{proof}
The proof is given as follows. Based on the Eqs. (\ref{eq13}) and (\ref{eq15}), one will have
\begin{equation}
\begin{array}{l}
S^2(\rho_{12})+\tau_{13}^2+\tau_{23}^2\\
=1+8\alpha_0^2(\alpha_{3}^2+\alpha_4^2)\\
=1+8\alpha_0^2(1-\alpha_0^2-\alpha^2_1-\alpha_2^2)\\
=3-8(\alpha_0^2-\frac{1}{2})^2-8\alpha_0^2\alpha^2_1-8\alpha_0^2\alpha^2_2\\
\leq3.
\end{array}
\end{equation}
Combining the above inequality and $\tau_{13}=3f(\rho_{13})-2$, $\tau_{23}=3f(\rho_{23})-2$, the inequality $S^2(\rho_{12})+[3f(\rho_{13})-2]^2+[3f(\rho_{23})-2]^2\leq3$ is obtained. Similarly, one can get the other two inequalities $S^2(\rho_{13})+[3f(\rho_{12})-2]^2+[3f(\rho_{23})-2]^2\leq3$ and $S^2(\rho_{23})+[3f(\rho_{12})-2]^2+[3f(\rho_{13})-2]^2\leq3$. Therefore, the proof of the third result is completed.
\end{proof}

The result indicates that the steering observable of a reduced two-qubit state and the average teleportation fidelities of the other two reduced two-qubit state are upper bounded. Obviously, this is also a trade-off relation among the steerable observable and the average teleportation fidelities for three-qubit pure states.

%For the fully entangled three-qubit state $|\psi\rangle=(|100\rangle+|010\rangle+\sqrt{2}|001\rangle)/2$, the maximum value of the left hand of the Eq. (16), i.e., $S^2(\rho_{ij})+S^2(\rho_{ik})+[3f(\rho_{jk})-2]^2$, is $11/4$.
Now, we identify a class of genuinely entangled states  $|\phi\rangle_{123}=\sqrt{0.5}|000\rangle+\sqrt{0.5-q^2}|101\rangle+q|111\rangle$ with $q\in(0,\sqrt{0.5})$. For this class of states, $S(\rho_{12})=S(\rho_{23})=\sqrt{2}q$, $S(\rho_{13})=\sqrt{3-4q^2}$, $f(\rho_{12})=f(\rho_{23})=(2+\sqrt{2}q)/3$ and $f(\rho_{13})=1$. These results will saturate the inequalities given in Eqs. (16) and (18), i.e., $S^2(\rho_{12})+S^2(\rho_{13})+[3f(\rho_{23})-2]^2=3$, $S^2(\rho_{13})+S^2(\rho_{23})+[3f(\rho_{12})-2]^2=3$, and $S^2(\rho_{13})+[3f(\rho_{12})-2]^2+[3f(\rho_{23})-2]^2=3$.

\section{Conclusion and discussion}
In this paper, we consider the applications of EPR steering in quantum teleportation. The upper bound of the average teleportation fidelity based on the EPR steering is derived. When the receiver can only perform the identity or the Pauli rotation operations, the X-type states violating the three-setting linear steering inequality could always be useful for teleportation. For the three-qubit pure states, the steering observables and the average teleportation fidelities of the reduced two-qubit systems maintain the same ordering. Furthermore, the complementary relations among the steering observables of two reduced two-qubit states and the average teleportation fidelity of the third reduced two-qubit state, or among the steering observable of a reduced two-qubit state and the average teleportation fidelities of the other two reduced two-qubit states for the three-qubit pure states are presented.

There are steerable states which do not violate the three-setting linear steering inequality. Thus, it is an open question to give an appropriate steering criterion and then consider the applications of EPR steering in quantum teleportation. Furthermore, as EPR steering is an important quantum resource in quantum information and computation, it will also be interesting to investigate its applications in other quantum information tasks.

\section*{ACKNOWLEDGEMENTS} The work was supported by the Fundamental Research Funds for the Central Universities under Grant No. 2020ZDPYMS03.

\end{document}